\documentclass[%
aip,
jmp,%
amsmath,amssymb,
reprint,%
]{revtex4-1}

\usepackage{color}
\usepackage{yfonts}
\usepackage[pdftex]{graphicx}
\setkeys{Gin}{clip=true,draft=false}
\DeclareGraphicsExtensions{.pdf}
\usepackage{amsthm}
\usepackage{amsfonts}
\usepackage{amssymb}
\usepackage[cspex,bbgreekl]{mathbbol}
\usepackage[hang]{subfigure}
\usepackage{float}
\usepackage[super]{nth}
\usepackage{multirow} 
\usepackage{tabulary}
\newcolumntype{K}[1]{>{\centering\arraybackslash}p{#1}}
\newcommand{\xvec}{\mathbf{x}}

\newtheoremstyle{definition_new}
{0.7\topsep}
{0.7\topsep}
{\normalfont}
{0pt}
{\itshape}
{}
{7pt}
{}
\theoremstyle{definition_new}

\newtheorem{thm}{Theorem}

\theoremstyle{remark}

\usepackage{graphicx}
\usepackage{dcolumn}
\usepackage{bm}

\begin{document}
	
	
\title{Absence of chaos in Digital Memcomputing Machines with solutions}
	
	\author{Massimiliano Di Ventra}
	\affiliation{Department of Physics, University of California, San Diego, La Jolla, 92093 CA}
	
	\author{Fabio L. Traversa}
	\affiliation{Department of Physics, University of California, San Diego, La Jolla, 92093 CA}
	\affiliation{LoGate Computing, Inc.}

\begin{abstract}
Digital memcomputing machines (DMMs) are non-linear dynamical systems designed so that their equilibrium points are solutions of the Boolean problem they solve. 
In a previous work [Chaos {\bf 27}, 023107 (2017)] it was argued that when DMMs support solutions of the associated Boolean problem then strange attractors cannot coexist with 
such equilibria. In this work, we demonstrate such conjecture. In particular, we show that both topological transitivity and the strongest property of topological mixing are inconsistent with the point dissipative property of DMMs when equilibrium points are present.  This is true for both the whole phase space and the global attractor. Absence of topological transitivity is enough to imply 
absence of chaotic behavior. In a similar vein, we prove that if DMMs do not have equilibrium points, the only attractors present are invariant tori/periodic orbits with periods that may possibly increase with system size (quasi-attractors). 
\end{abstract}

\maketitle




	

Memcomputing is a recently suggested computing paradigm that employs memory to both store and process information on the same physical location~\cite{diventra13a,UMM,traversaNP,DMM2}. It is based on the mathematical notion of 
Universal Memcomputing Machines introduced by us in Ref~\onlinecite{UMM}, which shows that these machines have the same computational power of non-deterministic Turing machines. As such, they can, in principle, solve Non-deterministic polynomial (NP)-complete problems with polynomial resources. 

In a subsequent work~\cite{DMM2}, we have shown that a physical and scalable realization of such machines --we call this subclass Digital Memcomputing Machines (DMMs)-- can be implemented using self-organizing logic gates (SOLGs). 
Unlike standard Boolean gates, SOLGs satisfy their logic proposition by also working ``in reverse'', in the sense that they are able to {\it adapt} to ``input'' signals from {\it any} terminal, without the distinction between input and output terminals as in conventional logic. 

From a physical point of view, these gates can be made as combination of standard circuit elements (including transistors) and/or memory elements such as memristors, memcapacitors and meminductors~\cite{09_memelements}. 
Irrespective of their practical realization (which incidentally is not even unique given the Boolean problem to solve) SOLGs and their combination into networks --self-organizing logic circuits-- are nothing other than non-linear 
dynamical systems, namely they can be mathematically described by an equation of motion of the type:
\begin{equation}
\dot \xvec(t) = F(\xvec(t)) \label{ODE},
\end{equation}
where $\xvec\in X $ is the set of state variables (such as voltages, currents and internal memory variables) belonging to the phase space $X$, and $F$ is the flow vector field representing the laws of temporal evolution of $\xvec$. 
This way, the solutions of the problem at hand can be found by driving the corresponding dynamical system to its equilibrium points. 

However, DMMs possess an additional, important feature that is not common to all dynamical systems described by Eq.~(\ref{ODE}): DMMs are {\it point dissipative}~\cite{hale_2010_asymptotic}. This mathematical property implies that there is 
always a bounded and compact set $J\subset X$ that attracts {\it every} point of $X$ under $F$.~\cite{hale_2010_asymptotic} This is indeed a critical feature for DMMs to succeed in solving complex problems efficiently since it implies that, irrespective of initial conditions, {\it all} system's orbits will eventually enter the attractor $J$ and remain there. 

However, such attractor could be the 
union of equilibrium points (the solutions of the Boolean problem to be solved), invariant tori (that include periodic and quasi-periodic orbits), and strange attractors. Formally, if $B_i$'s are the subsets of $J$ that attract the different solutions (stable fixed points, invariant tori and 
strange attractors), we can write $J=\bigcup_{B_i}\omega(B_i) $, with the $\omega$-limit sets $\omega(B_i)$ of $B_i$ the sets $\omega(B_i)={\textstyle\bigcap_{s\geq0}}\mathrm{Cl}{\textstyle\bigcup_{t\geq s}}T(t)B_i$ with $T:\mathbb{R}\times X\rightarrow X$, 
\begin{equation}
T(t)\xvec(0)=\xvec(0)+\int_0^tF(\xvec(t^{\prime}))dt^{\prime} \label{cr_semigroup}
\end{equation}
a $C^r$ semigroup with $r\geq 1$, and Lipschitz continuous. The fact that $J=\bigcup_{B_i}\omega(B_i) $ is a consequence of $J$ being uniformly asymptotically stable (namely it is stable and attracts a neighborhood of J)~\cite{hale_2010_asymptotic}. In addition, the Lipschitz continuity implies that $\omega(B_i) \cap \omega(B_j) = \emptyset$, $\forall B_i, B_j \subset X$. Note also that the phase space, $X$, of DMMs is just a metric compact space, not a Banach (vector) space. This means that Eq.~(\ref{ODE}) does not  
always have a fixed point~\cite{hale_2010_asymptotic}. 

In Ref.~\onlinecite{DMM2}, we have provided physical and mathematical arguments as to why, in the presence of equilibrium points, neither invariant tori nor strange attractors can coexist with equilibria. This conjecture is also supported by numerical simulations of DMMs. Our argument was that since the equilibrium points of DMMs belong to an open ball with finite radius in which the convergence rate to equilibrium is exponential, then if there are any orbits that evolve into a strange attractor or a limit cycle, they would eventually enter that ball and converge exponentially fast to one equilibrium point. 
However, we could not prove these statements rigorously. In this paper, we prove one of the above two statements, the one concerning chaos. 

Despite considerable research on chaotic dynamics, there is yet no agreed-upon definition of what constitutes chaos, so in order to prove its 
absence, one would be forced to prove the absence of most of its features~\cite{Igor}. However, there is a fairly general understanding, in both the mathematics and science 
communities, that deterministic chaotic dynamics 
is equivalent to these three properties~\cite{Devaney,newpaper} that need to be satisfied by an invariant subset $A$ of $X$ ($T(A)=A$): 
\emph{(i)} $T$ is {\it topologically transitive} in $A$,
\emph{(ii)} the set of periodic orbits of $T$ is dense in $A$, and \emph{(iii)} $T|_A$ depends sensitively on initial conditions. Topological transitivity means that, given any two points in the phase space, we can find an orbit that comes arbitrarily close to both. 

Mathematically, we say that the flow is topologically transitive~\cite{book-top-dyn} 
if for any pair of non-vacuous open sets $U$, $V$ $\subset A$, there exists a positive time $t'$ such that $T|_A(t',U)\cap V \neq \emptyset$, where $T|_A(t',U)=\{t' \in \mathbb{R}^+; \xvec(0) \in U|\; T|_A(t')\xvec(0):U\to A \}$~\footnote{For maps this is replaced by the requirement that there exists a positive $n$ such that $T|_A^{n}(U)\cap V \neq \emptyset$, where $T_A^n$ is the $n-$th iterate of the map restricted on $A$.}.  Property \emph{(iii)} is redundant since one can prove~\cite{Banks} that if the dynamical system satisfies \emph{(i)} and \emph{(ii)}, then \emph{(iii)} follows~\footnote{Since DMMs are not maps on intervals of $\mathbb{R}$, transitivity does not imply a dense set of periodic orbits~\cite{Vellekoop}.}. Incidentally, this means that only topological (global) features are enough to characterize a flow as 
chaotic, without the need of introducing a metric (local property). 

A stronger condition than transitivity is {\it topological mixing} which states that any open set of $A \subset X$ overlaps with any other open set if the former is propagated long enough. Mathematically, this means that for any pair of non-vacuous open sets $U$, $V$ $\subset A$, there exists a positive time $t'$ such that $T|_A(t\geq t',U)\cap V \neq \emptyset$, where $T|_A(t\geq t',U)=\{t'\leq t \in \mathbb{R}^+; \xvec(t') \in U| \; T|_A(t)\xvec(t'):U\to A \}$~\footnote{For maps this is replaced by the requirement that there exists a positive $n$ such that $T|_A^{m}(U)\cap V \neq \emptyset$, $\forall m\geq n$.}. Topological mixing implies transitivity, but not the other way around~\cite{Katok}. Note also that if $A=X$ then chaos spans the whole 
phase space.

If we take at face value the above properties as definitions of chaos, it is then sufficient to prove the absence of either {\it (i)} or {\it (ii)} to prove the original statement for our DMMs. We will focus on property {\it (i)}, namely topological transitivity. We can also go a step further and prove absence of 
topological mixing in DMMs with equilibrium points. We will then discuss which invariant dynamical properties DMMs show if there are no solutions to the Boolean problems they solve. 

{\it Absence of chaos with equilibrium points --} We consider a DMM described by the dynamical system $(X,T)$, with $X$ a metric compact space of dimension $D=\dim X$ and $T$ given in Eq.~(\ref{cr_semigroup}). We do not need to fix the topology of this space, but without loss of generality we may choose the Hausdorff topology. 
Our aim is then to prove the following: 
\begin{thm}
	\label{toptransitive} If $(X,T)$ is point dissipative and it has at least one equilibrium point, then it is not topologically transitive, and hence it is not topologically mixing, both in $X$ and in the global attractor $J$. 
\end{thm}

\begin{proof}
	The proof is simple. In fact, it is known that a dynamical system with a closed invariant set $A$ with nonempty interior, and different from the whole space $X$ cannot be transitive~\cite{compact}. If the dynamical system is point dissipative and has at least one equilibrium point then such a closed invariant set exists 
	--the $\omega$-limit set of that point--, it is not empty and cannot be equal to the whole space $X$: $A\subseteq J\subset X$. Since topological mixing 
	implies topological transitivity, the absence of the latter implies absence of the former. This result is valid even if we restrict the base space of the dynamical system to the global attractor $J$. The latter is compact, hence closed (and bounded). In the presence of an equilibrium point the (closed) $\omega$-limit set of that point belongs to $J$ and is not empty. Hence, either $J$ contains only the 
    $\omega$-limit set of the equilibrium point, or it is the union of such a set and some other set. Either way, the dynamical system cannot be 
	topologically transitive in $J$, hence it cannot be mixing. 
\end{proof}



If we take these features (in particular transitivity) as indications of the absence of strange attractors (chaos), we have indeed proved the conjecture advanced in Ref.~\onlinecite{DMM2} that in DMMs with solutions chaos {\it cannot emerge} during dynamics. 

{\it Absence of chaos without equilibrium points --} Let us now consider the case in which DMMs do {\it not} support equilibrium points: the Boolean problems represented by DMMs have no solutions. Then the compact metric nature of the support of the corresponding dynamical system $(X,T)$ implies the existence of at least one recurrent point~\cite{Devaney}. Recurrent points, other than fixed points, are (high-dimensional) invariant tori~\cite{Devaney}. This means that DMMs with no solution support at least one invariant torus. If this is the case, we can then prove the following theorem:

\begin{thm}
	\label{limit-cycle} If $(X,T)$ is point dissipative and it does not have equilibrium points, then it is not topologically transitive, and it is not topologically mixing in both $X$ and in the global attractor $J$. 
\end{thm}

\begin{proof}
	The proof of this theorem is analogous to the one for Theorem~\ref{toptransitive} with the only provision that the global attractor $J\subset X$ is now the union of $\omega$-limit sets of only invariant tori (at least one, since the 
	space $X$ is metric compact) and, possibly, strange attractors. Again, if this is true and one of these $\omega$-limit sets is an invariant torus, then this is a closed invariant set with nonempty interior and different from the whole space. Therefore, topological transitivity cannot occur in $X$. This implies that the system cannot be topologically mixing on $X$. Similarly, also the restriction of the flow to the global attractor $J$ precludes topological transitivity, hence mixing in $J$, since the global attractor contains the $\omega$-limit set of the invariant torus. 
\end{proof}

In other words, if the DMMs do not support solutions, they cannot manifest strange attractors (chaotic behavior), only invariant tori. However, the type and structure of these cycles is not obvious and they may indeed have periods that vary (increase) with 
the dimension of the phase space. 

In fact, as the size of the phase space increases we cannot exclude the existence of {\it quasi-attractors}, namely stable periodic trajectories with extremely large periods and very small attraction basins~\cite{Shilnikov}. If these 
quasi-attractors exist, then their features may be practically indistinguishable from true chaotic behavior. The reason is that the system may, due to some small fluctuations, jump from one of these long-period cycles to another, thus rendering the dynamics practically chaotic~\cite{Synergetics}. 

{\it Conclusions --} In summary, we have proved a conjecture advanced by us in Ref.~\onlinecite{DMM2} that in the presence of solutions of a given Boolean problem solved by DMMs, chaotic behavior cannot coexist with such dynamics. This is an 
important consequence of the point dissipative property of the dynamical systems representing these machines. Point dissipative means that there is always a bounded compact subset of the whole phase space that attracts {\it any} orbit 
of the system. This global attractor then fully characterizes the dynamics at sufficiently long times. 

We have also demonstrated that, in view of the (additional) compact metric nature of the dynamical systems describing DMMs, in the absence of equilibrium points, (high-dimensional) invariant tori are the only possible attractors of the system. However, we cannot exclude that such invariant tori are quasi-attractors. If that is the case, then their period and number could make them indistinguishable from 
chaotic behavior in numerical or experimental situations. 

In order to complete the theoretical program initiated in Ref.~\onlinecite{DMM2} 
we still need to prove the other conjecture: that in the presence of equilibrium points DMMs cannot support periodic orbits. We leave this (difficult) open problem for future work.

{\it Acknowledgments --} M.D. and F.L.T. acknowledge partial support from the Center for Memory Recording Research at UCSD and LoGate Computing, Inc.

\bibliographystyle{apsrev4-1}
\bibliography{SUSYref}

\begin{thebibliography}{19}%
\makeatletter
\providecommand \@ifxundefined [1]{%
 \@ifx{#1\undefined}
}%
\providecommand \@ifnum [1]{%
 \ifnum #1\expandafter \@firstoftwo
 \else \expandafter \@secondoftwo
 \fi
}%
\providecommand \@ifx [1]{%
 \ifx #1\expandafter \@firstoftwo
 \else \expandafter \@secondoftwo
 \fi
}%
\providecommand \natexlab [1]{#1}%
\providecommand \enquote  [1]{``#1''}%
\providecommand \bibnamefont  [1]{#1}%
\providecommand \bibfnamefont [1]{#1}%
\providecommand \citenamefont [1]{#1}%
\providecommand \href@noop [0]{\@secondoftwo}%
\providecommand \href [0]{\begingroup \@sanitize@url \@href}%
\providecommand \@href[1]{\@@startlink{#1}\@@href}%
\providecommand \@@href[1]{\endgroup#1\@@endlink}%
\providecommand \@sanitize@url [0]{\catcode `\\12\catcode `\$12\catcode
  `\&12\catcode `\#12\catcode `\^12\catcode `\_12\catcode `\%12\relax}%
\providecommand \@@startlink[1]{}%
\providecommand \@@endlink[0]{}%
\providecommand \url  [0]{\begingroup\@sanitize@url \@url }%
\providecommand \@url [1]{\endgroup\@href {#1}{\urlprefix }}%
\providecommand \urlprefix  [0]{URL }%
\providecommand \Eprint [0]{\href }%
\providecommand \doibase [0]{http://dx.doi.org/}%
\providecommand \selectlanguage [0]{\@gobble}%
\providecommand \bibinfo  [0]{\@secondoftwo}%
\providecommand \bibfield  [0]{\@secondoftwo}%
\providecommand \translation [1]{[#1]}%
\providecommand \BibitemOpen [0]{}%
\providecommand \bibitemStop [0]{}%
\providecommand \bibitemNoStop [0]{.\EOS\space}%
\providecommand \EOS [0]{\spacefactor3000\relax}%
\providecommand \BibitemShut  [1]{\csname bibitem#1\endcsname}%
\let\auto@bib@innerbib\@empty
\bibitem [{\citenamefont {Di~Ventra}\ and\ \citenamefont
  {Pershin}(2013)}]{diventra13a}%
  \BibitemOpen
  \bibfield  {author} {\bibinfo {author} {\bibfnamefont {M.}~\bibnamefont
  {Di~Ventra}}\ and\ \bibinfo {author} {\bibfnamefont {Y.~V.}\ \bibnamefont
  {Pershin}},\ }\href@noop {} {\bibfield  {journal} {\bibinfo  {journal}
  {Nature Physics}\ }\textbf {\bibinfo {volume} {9}},\ \bibinfo {pages} {200}
  (\bibinfo {year} {2013})}\BibitemShut {NoStop}%
\bibitem [{\citenamefont {Traversa}\ and\ \citenamefont
  {Di~Ventra}(2015)}]{UMM}%
  \BibitemOpen
  \bibfield  {author} {\bibinfo {author} {\bibfnamefont {F.~L.}\ \bibnamefont
  {Traversa}}\ and\ \bibinfo {author} {\bibfnamefont {M.}~\bibnamefont
  {Di~Ventra}},\ }\href {\doibase 10.1109/TNNLS.2015.2391182} {\bibfield
  {journal} {\bibinfo  {journal} {IEEE Trans. Neural Netw. Learn. Syst.}\
  }\textbf {\bibinfo {volume} {26}},\ \bibinfo {pages} {2702} (\bibinfo {year}
  {2015})}\BibitemShut {NoStop}%
\bibitem [{\citenamefont {Traversa}\ \emph {et~al.}(2015)\citenamefont
  {Traversa}, \citenamefont {Ramella}, \citenamefont {Bonani},\ and\
  \citenamefont {Di~Ventra}}]{traversaNP}%
  \BibitemOpen
  \bibfield  {author} {\bibinfo {author} {\bibfnamefont {F.~L.}\ \bibnamefont
  {Traversa}}, \bibinfo {author} {\bibfnamefont {C.}~\bibnamefont {Ramella}},
  \bibinfo {author} {\bibfnamefont {F.}~\bibnamefont {Bonani}}, \ and\ \bibinfo
  {author} {\bibfnamefont {M.}~\bibnamefont {Di~Ventra}},\ }\href {\doibase
  10.1126/sciadv.1500031} {\bibfield  {journal} {\bibinfo  {journal} {Science
  Advances}\ }\textbf {\bibinfo {volume} {1}},\ \bibinfo {pages} {e1500031}
  (\bibinfo {year} {2015})}\BibitemShut {NoStop}%
\bibitem [{\citenamefont {Traversa}\ and\ \citenamefont
  {Di~Ventra}(2017)}]{DMM2}%
  \BibitemOpen
  \bibfield  {author} {\bibinfo {author} {\bibfnamefont {F.~L.}\ \bibnamefont
  {Traversa}}\ and\ \bibinfo {author} {\bibfnamefont {M.}~\bibnamefont
  {Di~Ventra}},\ }\href@noop {} {\bibfield  {journal} {\bibinfo  {journal}
  {Chaos: An Interdisciplinary Journal of Nonlinear Science}\ }\textbf
  {\bibinfo {volume} {27}},\ \bibinfo {pages} {023107} (\bibinfo {year}
  {2017})}\BibitemShut {NoStop}%
\bibitem [{\citenamefont {Di~Ventra}\ \emph {et~al.}(2009)\citenamefont
  {Di~Ventra}, \citenamefont {Pershin},\ and\ \citenamefont
  {Chua}}]{09_memelements}%
  \BibitemOpen
  \bibfield  {author} {\bibinfo {author} {\bibfnamefont {M.}~\bibnamefont
  {Di~Ventra}}, \bibinfo {author} {\bibfnamefont {Y.}~\bibnamefont {Pershin}},
  \ and\ \bibinfo {author} {\bibfnamefont {L.}~\bibnamefont {Chua}},\ }\href
  {\doibase 10.1109/JPROC.2009.2021077} {\bibfield  {journal} {\bibinfo
  {journal} {Proceedings of the IEEE}\ }\textbf {\bibinfo {volume} {97}},\
  \bibinfo {pages} {1717} (\bibinfo {year} {2009})}\BibitemShut {NoStop}%
\bibitem [{\citenamefont {Hale}(2010)}]{hale_2010_asymptotic}%
  \BibitemOpen
  \bibfield  {author} {\bibinfo {author} {\bibfnamefont {J.}~\bibnamefont
  {Hale}},\ }\href@noop {} {\emph {\bibinfo {title} {Asymptotic Behavior of
  Dissipative Systems}}},\ \bibinfo {edition} {2nd}\ ed.,\ \bibinfo {series}
  {Mathematical Surveys and Monographs}, Vol.~\bibinfo {volume} {25}\ (\bibinfo
   {publisher} {American Mathematical Society},\ \bibinfo {address}
  {Providence, Rhode Island},\ \bibinfo {year} {2010})\BibitemShut {NoStop}%
\bibitem [{\citenamefont {Ovchinnikov}\ and\ \citenamefont
  {Di~Ventra}(2017)}]{Igor}%
  \BibitemOpen
  \bibfield  {author} {\bibinfo {author} {\bibfnamefont {I.}~\bibnamefont
  {Ovchinnikov}}\ and\ \bibinfo {author} {\bibfnamefont {M.}~\bibnamefont
  {Di~Ventra}},\ }\href@noop {} {\bibfield  {journal} {\bibinfo  {journal}
  {arxiv:1702.06561}\ } (\bibinfo {year} {2017})}\BibitemShut {NoStop}%
\bibitem [{\citenamefont {Devaney}(1992)}]{Devaney}%
  \BibitemOpen
  \bibfield  {author} {\bibinfo {author} {\bibfnamefont {R.}~\bibnamefont
  {Devaney}},\ }\href@noop {} {\emph {\bibinfo {title} {A First Course in
  Chaotic Dynamical Systems: Theory and Experiment}}}\ (\bibinfo  {publisher}
  {Addison-Wesley},\ \bibinfo {year} {1992})\BibitemShut {NoStop}%
\bibitem [{\citenamefont {Aulbach}\ and\ \citenamefont
  {Kieninger}(2001)}]{newpaper}%
  \BibitemOpen
  \bibfield  {author} {\bibinfo {author} {\bibfnamefont {B.}~\bibnamefont
  {Aulbach}}\ and\ \bibinfo {author} {\bibfnamefont {B.}~\bibnamefont
  {Kieninger}},\ }\href@noop {} {\bibfield  {journal} {\bibinfo  {journal}
  {Nonlinear Dynamics and Systems Theory}\ }\textbf {\bibinfo {volume} {1}},\
  \bibinfo {pages} {23–37} (\bibinfo {year} {2001})}\BibitemShut {NoStop}%
\bibitem [{\citenamefont {Gottschalk}\ and\ \citenamefont
  {Hedlund}(1955)}]{book-top-dyn}%
  \BibitemOpen
  \bibfield  {author} {\bibinfo {author} {\bibfnamefont {W.}~\bibnamefont
  {Gottschalk}}\ and\ \bibinfo {author} {\bibfnamefont {G.}~\bibnamefont
  {Hedlund}},\ }\href@noop {} {\emph {\bibinfo {title} {Topological
  Dynamics}}}\ (\bibinfo  {publisher} {American Mathematical Society},\
  \bibinfo {year} {1955})\BibitemShut {NoStop}%
\bibitem [{Note1()}]{Note1}%
  \BibitemOpen
  \bibinfo {note} {For maps this is replaced by the requirement that there
  exists a positive $n$ such that $T|_A^{n}(U)\cap V \not =\emptyset $, where
  $T_A^n$ is the $n-$th iterate of the map restricted on $A$.}\BibitemShut
  {Stop}%
\bibitem [{\citenamefont {Banks}\ \emph {et~al.}(1992)\citenamefont {Banks},
  \citenamefont {Brooks}, \citenamefont {Cairns}, \citenamefont {Davis},\ and\
  \citenamefont {Stacey}}]{Banks}%
  \BibitemOpen
  \bibfield  {author} {\bibinfo {author} {\bibfnamefont {J.}~\bibnamefont
  {Banks}}, \bibinfo {author} {\bibfnamefont {J.}~\bibnamefont {Brooks}},
  \bibinfo {author} {\bibfnamefont {G.}~\bibnamefont {Cairns}}, \bibinfo
  {author} {\bibfnamefont {G.}~\bibnamefont {Davis}}, \ and\ \bibinfo {author}
  {\bibfnamefont {P.}~\bibnamefont {Stacey}},\ }\href@noop {} {\bibfield
  {journal} {\bibinfo  {journal} {The American Mathematical Monthly}\ }\textbf
  {\bibinfo {volume} {99}},\ \bibinfo {pages} {332} (\bibinfo {year}
  {1992})}\BibitemShut {NoStop}%
\bibitem [{Note2()}]{Note2}%
  \BibitemOpen
  \bibinfo {note} {Since DMMs are not maps on intervals of $\protect \mathbb
  {R}$, transitivity does not imply a dense set of periodic orbits~\cite
  {Vellekoop}.}\BibitemShut {Stop}%
\bibitem [{Note3()}]{Note3}%
  \BibitemOpen
  \bibinfo {note} {For maps this is replaced by the requirement that there
  exists a positive $n$ such that $T|_A^{m}(U)\cap V \not =\emptyset $,
  $\forall m\geq n$.}\BibitemShut {Stop}%
\bibitem [{\citenamefont {Hasselblatt}\ and\ \citenamefont
  {Katok}(2003)}]{Katok}%
  \BibitemOpen
  \bibfield  {author} {\bibinfo {author} {\bibfnamefont {B.}~\bibnamefont
  {Hasselblatt}}\ and\ \bibinfo {author} {\bibfnamefont {A.}~\bibnamefont
  {Katok}},\ }\href@noop {} {\emph {\bibinfo {title} {A First Course in
  Dynamics}}}\ (\bibinfo  {publisher} {Cambridge University Press},\ \bibinfo
  {year} {2003})\BibitemShut {NoStop}%
\bibitem [{\citenamefont {Balibrea}\ and\ \citenamefont
  {Snoha}(2003)}]{compact}%
  \BibitemOpen
  \bibfield  {author} {\bibinfo {author} {\bibfnamefont {F.}~\bibnamefont
  {Balibrea}}\ and\ \bibinfo {author} {\bibfnamefont {L.}~\bibnamefont
  {Snoha}},\ }\href@noop {} {\bibfield  {journal} {\bibinfo  {journal}
  {Topology and its Applications}\ }\textbf {\bibinfo {volume} {133}},\
  \bibinfo {pages} {225} (\bibinfo {year} {2003})}\BibitemShut {NoStop}%
\bibitem [{\citenamefont {Afraimovich}\ and\ \citenamefont
  {Shilnikov}(1991)}]{Shilnikov}%
  \BibitemOpen
  \bibfield  {author} {\bibinfo {author} {\bibfnamefont {V.}~\bibnamefont
  {Afraimovich}}\ and\ \bibinfo {author} {\bibfnamefont {L.}~\bibnamefont
  {Shilnikov}},\ }\href@noop {} {\bibfield  {journal} {\bibinfo  {journal}
  {Amer. Math. Soc. Transl.}\ }\textbf {\bibinfo {volume} {149}},\ \bibinfo
  {pages} {201} (\bibinfo {year} {1991})}\BibitemShut {NoStop}%
\bibitem [{\citenamefont {Mikhailov}\ and\ \citenamefont
  {Loskutov}(1996)}]{Synergetics}%
  \BibitemOpen
  \bibfield  {author} {\bibinfo {author} {\bibfnamefont {A.}~\bibnamefont
  {Mikhailov}}\ and\ \bibinfo {author} {\bibfnamefont {A.~Y.}\ \bibnamefont
  {Loskutov}},\ }\href@noop {} {\emph {\bibinfo {title} {Foundations of
  Synergetics II: Chaos and Noise}}}\ (\bibinfo  {publisher} {Springer},\
  \bibinfo {year} {1996})\BibitemShut {NoStop}%
\bibitem [{\citenamefont {Vellekoop}\ and\ \citenamefont
  {Berglund}(1994)}]{Vellekoop}%
  \BibitemOpen
  \bibfield  {author} {\bibinfo {author} {\bibfnamefont {M.}~\bibnamefont
  {Vellekoop}}\ and\ \bibinfo {author} {\bibfnamefont {R.}~\bibnamefont
  {Berglund}},\ }\href@noop {} {\bibfield  {journal} {\bibinfo  {journal} {The
  American Mathematical Monthly}\ }\textbf {\bibinfo {volume} {101}},\ \bibinfo
  {pages} {353} (\bibinfo {year} {1994})}\BibitemShut {NoStop}%
\end{thebibliography}%
\end{document}